\documentclass{amc}
\usepackage{amsmath}
\usepackage{amsfonts,amssymb}
\usepackage{amscd}

\usepackage[colorlinks=true]{hyperref}
 \hypersetup{urlcolor=blue, citecolor=red}
\def\currentvolume{8}
 \def\currentissue{1}
  \def\currentyear{2014}
   
    \def\ppages{21--33}
     \def\DOI{10.3934/amc.2013.8.21}

\newtheorem{theorem}{Theorem}
\newtheorem{lemma}[theorem]{Lemma}
\newtheorem{proposition}[theorem]{Proposition}
\newtheorem{corollary}[theorem]{Corollary}
\theoremstyle{definition}
\newtheorem{definition}[theorem]{Definition}
\newtheorem*{defNotag}{Definition}
\newtheorem*{remark}{Remark}
\newtheorem*{comt}{Comment}
\newtheorem*{ex}{Example}
\newtheorem{exam}{Example}
\newtheorem*{notation}{Notation}
\newcommand{\field}[1]{\mathbb{#1}}
\newcommand{\F}{\field{F}}
\newcommand{\N}{\field{N}}

\newcommand{\Z}{\field{Z}}

\title[Special bent and near-bent functions]
      {Special bent and near-bent functions}

\author[J. Wolfmann]{}

\subjclass{xxXxx}


\begin{document}
\maketitle

\centerline{\scshape J. Wolfmann}
\medskip
{\footnotesize
  \centerline{IMATH(IAA), Universit\'e du Sud Toulon-Var}
   \centerline{83957 La Garde Cedex, France}
}

\bigskip

\centerline{(Communicated by Simon Litsyn)}

\begin{abstract}
Starting from special near-bent functions in dimension $2t-1$ we construct bent functions in dimension  $2t$
having  a specific derivative. We deduce new families of bent functions.
\end{abstract}

\section{Introduction}\label{sec1}

Let $\F_{2}$ be the finite field of order 2. An $m$-boolean function (or boolean function in $m$ dimensions) is a map $F$ from  $\F_{2}^{m}$ to $\F_{2}$.

Bent functions are the boolean functions whose Fourier coefficients have
constant magnitude. They were introduced by Rothaus in \cite{r8}. An $m$-boolean function $F$ is bent if  all its Fourier coefficients are in
$\lbrace -2^{m/2},2^{m/2}\rbrace$.
Bent functions are  of interest for Coding Theory, Cryptology and well-correlated binary
sequences. For example, they have the maximum possible Hamming distance to the set of affine boolean functions.
They were the topic of a lot of works (see \cite{r2,r3,r6,r7,r9,r10}) but the complete classification of bent functions
and other questions are still open.
\par
By definition, a $m$-boolean function $F$ is near-bent if all its Fourier coefficients are in
$\lbrace -2^{(m+1)/2},0,2^{(m+1)/2}\rbrace$. Since the Fourier coefficients are in $\Z$ the bent functions
in $m$ dimensions exist only when $m$ is even and near-bent functions in $m$ dimensions exist only when $m$ is odd.

\subsection{A two-variable representation}\label{sec1.1}

Assume $m=2t$. We identify  $\F_{2}^{2t}$ with  $\F_{2^{2t}}$ and $\F_{2^{2t}}$ with
$$
\F_{2^{2t-1}}\times \F_{2}=\lbrace X=(u,\nu)\mid u\in \F_{2^{2t-1}},\nu \in \F_{2}\rbrace.
$$
This decomposition can be used for a two-variable representation of $(2t)$-boolean functions as follows
$$
\F_{2t}^{(0)}=\lbrace (u,0)\mid u \in \F_{{ }_{2^{2t-1}}}\rbrace \quad{\rm and}\quad
\F_{2t}^{(1)}=\lbrace (u,1)\mid u \in \F_{{ }_{2^{2t-1}}}\rbrace.
$$
The restrictions of a  $(2t)$-boolean function $F$, respectively to $\F_{2t}^{(0)}$ and  to $\F_{2t}^{(1)}$ induce two
$(2t-1)$-boolean functions $f_{0}$ and $f_{1}$ defined by $f_{0}(u)=F(u,0)$ and $f_{1}(u)=F(u,1)$.

The two-variable representation (TVR) of $F$ is then defined by
\begin{equation}\label{eq1}
 \phi_{F}(x,y)=(y+1)f_{0}(x)+yf_{1}(x).
\end{equation}
Note that  $F(u,0)=f_{0}(u)=\phi_{F}(u,0)$ and $F(u,1)=f_{1}(u)=\phi_{F}(u,1)$. Hence if $X=(u,\nu)$ then $F(X)=\phi_{F}(u,\nu)$.

It can be proved that if $F$ is a $(2t)$-bent function then $f_{0}$ and $f_{1}$ are $(2t-1)$-near-bent functions (see  3.1.4).
This leads to considering the following inverse problem:
 construct $(2t)$-bent functions from $(2t-1)$-near-bent functions. This was the purpose of \cite{r7} and in a certain sense of \cite{r10} where connections
with cyclic codes were established.

Using the above representation of $\F_{2^{2t}}$, it could be easily checked that $f_{0}+f_{1}$ is the derivative of $F$ with respect
to $(0,1)$ (see next section for the definition of the derivative). A possible way to study bent functions is to classify them regarding the  degree (as a boolean function) of this derivative.
 The most simple case to consider is when this degree  is one.
In the present paper we are concerned with bent functions such that $f_{0}+f_{1}=tr+\xi$ where $tr$ is the trace of $\F_{2^{2t-1}}$ and
 $\xi\in \lbrace 0,1\rbrace$. We introduce the new notion of pseudo-duality and as application we present new families of bent functions.
 By the way we present a generalization of a result on the Gold function (Lemma \ref{lem15} in Section \ref{sec3.2}).

\section{Results}\label{sec2}
 The proofs of the Theorems of this section will be given in Section 3 while the other results are proven in the present section.

 First recall some usual definitions of a $m$-boolean function $F$.

$\bullet$ If $e\in \F_{2}^{m}$ then the Derivative of $F$ with respect to $e$ is the
$m$-boolean function $D_{e}F$ defined by
 $$D_{e}F(X)=F(X)+F(X+e).$$

\begin{remark}
In this paper we consider the following special cases.
 If $F$ is a $(2t)$-boolean function whose two-variable representation is given in \eqref{eq1} then direct calculation shows that
$$D_{(0,1)}(F)=f_{0}+f_{1}.$$
If $f$ is a $(2t-1)$-boolean function then
$$D_{1}f(x)=f(x+1)+f(x).$$
\end{remark}

$\bullet$ The Fourier transform (or the Walsh transform) $\hat{F}$ of  $F$ is the map from $\F_{2}^{m}$ into $\Z$
 defined by
$$\hat{F}(v)=\sum_{X\in\F_{2^{m}}}(-1)^{F(X)+ <v,X>},$$
where  $< ,>$ denotes any inner product of $\F_{2^{m}}$ over $\F_{2}$.
$\hat{F}(v)$ is called  the Fourier coefficient of $v$.

\begin{remark}
The set of $\hat{F}(v)$ when $v$ runs through $\F_{2^{m}}$ is independent of the choice of the inner product $<,>$.
\end{remark}

$\bullet$ $F$ is bent if  all its Fourier coefficients are in $\lbrace -2^{m/2},2^{m/2}\rbrace$.
$F$ is near-bent if all its Fourier coefficients are in $\lbrace -2^{(m+1)/2},0,2^{(m+1)/2}\rbrace$.

$\bullet$ If $m=2t$ and if $F$ is a bent function then the dual  $\tilde{F}$ of   $F$ is the $(2t)$-boolean function defined by
$$\hat{F}(v)=(-1)^{\tilde{F}(v)}2^{t},$$
where $\hat{F}$ is the Fourier transform of $F$.
It is well-known, and easy to proof  that the dual of a bent function is a bent function (see \cite{r3} or \cite{r8}).

$\bullet$ The algebraic degree, or more simply ``the degree'' of an $m$-boolean function is the degree of its $m$-variable polynomial representation.
If it is expressed as $Tr(\pi(X))$ where $Tr$ is the trace function of $\F_{2^m}$ then its degree is the maximum of the binary weight of the monomial exponents of $\pi(X)$.
Recall that the binary weight weight of an integer is the number of non-zero coefficients of the binary expansion of this integer.

See References section for details on the previous definitions and results.

\subsection{Main theorems}\label{sec2.1}
From now on, we use the definitions and notations of the introduction. Thus  $\F_{2}^{2t}$ is identified with $\F_{2^{2t}}$ and $\F_{2^{2t}}$ with
 $\F_{2^{2t-1}}\times \F_{2}$. The two-variable representation (TVR) of a $(2t)$-boolean function $F$ is defined by
$\phi_{F}(x,y)=(y+1)f_{0}(x)+yf_{1}(x)$ where $f_{0}$ and  $f_{1}$ are defined in Section 1.
Obviously, $\phi_{F_{1}+F_{2}}=\phi_{F_{1}}+\phi_{F_{2}}$.

\begin{defNotag}
$f_{0}$ and $f_{1}$ are called the components of $F$.
\end{defNotag}

\begin{notation}
The trace of  $\F_{2^{2t-1}}$ is denoted by $tr$ and defined by
$$tr(x)=\sum_{i=0}^{2t-2}x^{2^{i}}.$$
\end{notation}

The next theorem sets a condition on a near-bent function to be the first component of a bent function.
\begin{theorem}\label{thm1}
Let $f_{0}$ be a $(2t-1)$-near-bent function. If the derivative $D_{1}f_{0}$ is a constant function
then the $(2t)$-boolean function $F$ such that
$$
\phi_{F}(x,y)=(y+1)f_{0}(x)+yf_{1}(x) \quad{\rm with}\quad f_{0}+f_{1}=tr
$$
is a bent function.
\end{theorem}

We now present a pseudo-reciprocoal theorem.

\begin{theorem}\label{thm2}
Let $F$ be a bent function and $\tilde{F}$ be its dual bent function with
$$
\phi_{F}(x,y)=(y+1)f_{0}(x)+yf_{1}(x) \quad{\rm and} \phi_{\tilde{F}}(x,y)=(y+1)\tilde{f}_{0}(x)+y\tilde{f}_{1}(x).
$$
If $f_{0}+f_{1}=tr$, then $D_{1}\tilde{f}_{0}=0$ and $D_{1}\tilde{f}_{1}=1$.
\end{theorem}

The converse of Theorem \ref{thm1} is not true. In other words, it is not true that if $f_{0}+f_{1}=tr$ then $D_{1}f_{0}$ is a constant function,
 as it will be seen in Example 2. However, there is a special case given by the next corollary which follows  immediately from Theorem \ref{thm2}.

\begin{corollary}\label{cor3}
With the above notation,  let $F$ be a bent  such that $f_{0}+f_{1}=tr$.
If $F$ is self dual, say $F=\tilde{F}$, then $D_{1}{f}_{0}=0$.
\end{corollary}

Relations between the components of the dual $\tilde{F}$ of a bent function $F$ such that  $f_{0}+f_{1}=tr$ are now presented in the next result.

\begin{theorem}\label{thm4}
Let $F$ be a bent function and let $\tilde{F}$ be its dual bent function with
$$\phi_{F}(x,y)=(y+1)f_{0}(x)+yf_{1}(x) \quad{\rm and}\quad \phi_{\tilde{F}}(x,y)=(y+1)\tilde{f}_{0}(x)+y\tilde{f}_{1}(x).$$
Assume $f_{0}+f_{1}=tr$. Let $\hat{f}_{0}$ be the Fourier transform of $f_{0}$. Define
 $\mathcal{S}=\lbrace v\in \F_{2^{2t-1}}\mid \hat{f}_{0}(v)=-2^t\rbrace$ and $\mathcal{S}_{1}=\lbrace u+1\mid u\in \mathcal{S}\rbrace$.
Let $\mathcal{G}=\lbrace v\in \F_{2^{2t-1}}\mid \hat{f}_{0}(v)=0\rbrace$.
Let $g$ be the characteristic function of $\mathcal{G}$.
\begin{enumerate}
  \item[1)] The support of $\tilde{f}_{0}$ is $\mathcal{S}\cup \mathcal{S}_{1}$;
  \item[2)] $\tilde{f}_{1}(x) = \tilde{f}_{0}(x)+g(x)$.
\end{enumerate}
\end{theorem}

The next theorem states properties of  a bent function in the case when the hypothesis of Theorems \ref{thm1} and \ref{thm2} are both satisfied.

\begin{theorem}\label{thm5}
Let $H$ be a bent function and let $\tilde{H}$ be its dual bent function with
 $$\phi_{H}(x,y)=(y+1)h_{0}(x)+yh_{1}(x) \quad{\rm and}\quad \phi_{\tilde{H}}(x,y)=(y+1)\tilde{h}_{0}(x)+y\tilde{h}_{1}(x).$$
Assume $h_{0}+h_{1}=tr$.
\begin{enumerate}
  \item[a)] If $D_{1}h_{0}=0$ then $\tilde{h}_{0}+\tilde{h}_{1}=tr$;
  \item[b)] If $D_{1}h_{0}=1$ then $\tilde{h}_{0}+\tilde{h}_{1}=tr+1$.
\end{enumerate}
\end{theorem}

\subsection{Pseudo-duality}\label{sec2.2}
The previous results lead to the introduction of a new definition (notations are above).

\begin{definition}\label{def6}
 Let $G$ be a $(2t)$-bent function and let $\tilde{G}$ be its dual bent function with
$\phi_{\tilde{G}}(x,y)=(y+1)\tilde{g}_{0}(x)+y\tilde{g}_{1}(x)$.
The Pseudo-duals  of $G$ are the two $(2t)$-boolean function $\bar{G}_{0}$ and $\bar{G}_{1}$ defined by
\begin{align*}
&\phi_{\bar{G}_{0}}(x,y)=(y+1)\tilde{g}_{0}(x)+y(\tilde{g}_{0}(x)+tr(x));\\
&\phi_{\bar{G}_{1}}(x,y)=(y+1)\tilde{g}_{1}(x)+y(\tilde{g}_{1}(x)+tr(x)).
\end{align*}
\end{definition}

The meaning of this definition is given by the next theorem

\begin{theorem}\label{thm7}
Define the following two conditions on a $(2t)$-bent functions $\mathcal{G}$ with
$\phi_{G}(x,y)=(y+1)g_{0}(x)+yg_{1}(x)$.
\begin{enumerate}
  \item[$(\mathcal{T})$] $g_{0}+g_{1}=tr+\xi$ with $\xi\in \lbrace 0,1\rbrace$;
  \item[$(\mathcal{C})$] $D_{1}g_{0}=0$.
\end{enumerate}
If $F$ is a $(2t)$-bent function meeting condition $(\mathcal{T})$ then
\begin{enumerate}
  \item[A)] The pseudo-duals $\bar{F}_{0}$ and $\bar{F}_{1}$ are bent functions;
  \item[B)] The dual $\tilde{\bar{F}}_{0}$ of $\bar{F}_{0}$ meets $(\mathcal{C})$ and $(\mathcal{T})$ with  $\xi=0$;
  \item[C)] The dual $\tilde{\bar{F}}_{1}$ of $\bar{F}_{1}$ meets $(\mathcal{C})$ and $(\mathcal{T})$ with  $\xi=1$.
\end{enumerate}
\end{theorem}

\subsection{New families of bent functions}\label{sec2.3}
Let $\mathcal{F}$ be a family of $(2t)$-boolean functions, define $\bar{\mathcal{F}}_{0}=\lbrace \bar{F}_{0}\mid F\in \mathcal{F}\rbrace$
and $\bar{\mathcal{F}}_{1}=\lbrace \bar{F}_{1}\mid F\in \mathcal{F}\rbrace$.

By applying the previous theorem, if $\mathcal{F}$ is  a family of bent functions meeting $(\mathcal{T})$, then
$\bar{\mathcal{F}}_{0}$ and $\bar{\mathcal{F}}_{1}$ are new families of bent functions.

This is the case in the next proposition by using a family introduced in \cite[Theorem 9]{r7}.

\subsubsection{The Kasami-Welch example}\label{sec2.3.1}
This definition comes from the description of the near-bent function $f_{0}$ introduced in this example which is a classical object in the theory of boolean functions.

\begin{proposition}\label{prop8}
Let $t,s,d$ be integers such that
\begin{itemize}
  \item $2t-1$ is not divisible by $3$, $3s\equiv \pm 1 \mod (2t-1)$;
  \item $s<t$,  $d=4^s-2^s+1$.
\end{itemize}
Let $F$ be  the $(2t-1)$-boolean function  with
$$
\phi_{F}(x,y)=(y+1)tr(x^d)+ytr(x^d+x).
$$
The assertions  A), B), C) of Theorem \ref{thm7} hold for $F$.
\end{proposition}

\begin{proof}
In \cite[Theorem 9]{r7} it is proved that  $F$ is a bent function. Since $F$ satisfies  $(\mathcal{T})$
then  Theorem \ref{thm7} applies and gives the expected result.
\end{proof}

\begin{remark}
It is proved in \cite{r4} that  in the Kasami-Welch case, if $g$ is the function introduced in Theorem \ref{thm4} then $g(x)=1+tr(x^{2^{s}+1})$.
\end{remark}

\begin{ex}
$t=4,\,s=2$.
\medskip\par\noindent
$F:\quad f_{0}(x)=tr(x^{13}),\,f_{1}(x)=f_{0}(x)+tr(x)$,
\medskip\par\noindent
$\tilde{F}:\quad \tilde{f}_{0}(x)=tr(x^7+x^{11}+x^{19}+x^{21}),\,\tilde{f}_{1}(x)=\tilde{f}_{0}(x)+tr(x^5+1)$.

New bent functions
\medskip\par\noindent
$\bar{F}_{0}:\quad \bar{f}\,_{0}^{(0)}(x)=\tilde{f}_{0}(x),\,\bar{f}\,_{1}^{(0)}(x)=\tilde{f}_{0}(x)+tr(x)$,
\medskip\par\noindent
$\bar{F}_{1}:\quad \bar{f}\,_{0}^{(1)}(x)=\tilde{f}_{1}(x),\,\bar{f}\,_{1}^{(1)}(x)=\tilde{f}_{1}(x)+tr(x)$,
\medskip\par\noindent
$\tilde{\bar{F}}_{0}:\quad \tilde{\bar{f}}\,_{0}^{(0)}(x)=tr(x+x^{3}+x^{7}+x^{11}+x^{19}+x^{21})$,
\\
\hspace*{1cm} $\tilde{\bar{f}}\,_{1}^{(0)}(x)=\tilde{\bar{f}}\,_{0}^{(0)}(x)+tr(x)$,
\medskip\par\noindent
$\tilde{\bar{F}}_{1}:\quad \tilde{\bar{f}}\,_{0}^{(1)}(x)=tr(1+x^5+x^7+x^9+x^{11}+x^{19}+x^{21})$,
\\
\hspace*{1cm} $\tilde{\bar{f}}\,_{1}^{(1)}(x)=\tilde{\bar{f}}\,_{0}^{(1)}(x)+tr(x+1)$.
\end{ex}

\begin{remark}
In the above example $D_{1}f_{0}$ is not a constant function.
\end{remark}

\subsubsection{The quadratic case}\label{sec2.3.2}
In the next proposition we study the case where $f_{0}$ is quadratic and  such that $f_{0}(x)=tr(\pi(x))$
where all the coefficients of  $\pi(x)$ are in $\F_{2}$.

\begin{proposition}\label{prop9}
Let $f_{0}$ be a $(2t-1)$-near-bent function such that
\begin{itemize}
  \item $f_{0}(x)=tr(\pi(x))$ with $\pi(x)\in \F_{2}\lbrack x \rbrack$;
  \item The degree of $f_{0}$  is $2$.
\end{itemize}
\begin{enumerate}
  \item[1)] Then the $(2t)$-boolean function $F$ such that\\
            $$
            \phi_{F}(x,y)=(y+1)f_{0}(x)+yf_{1}(x) \quad{\rm with}\quad  f_{1}(x)=f_{0}(x)+tr(x)
            $$
            is a bent function.
  \item[2)] The dual $\tilde{F}$ of $F$ meets $(\mathcal{T})$.
  \item[3)] The  assertions  A), B), C) of Theorem \ref{thm7} hold for $F$.
\end{enumerate}
\end{proposition}

\begin{proof}
1) First notice that $r$ and $s$ are in $\N$, then for a convenient exponent $l$,
$$
tr(x^{2^{r}}+x^{2^{s}})=tr\lbrack(x^{2^{r}}+x^{2^{s}})^{2^{l}}\rbrack=tr(x^{2^{j}+1})
$$
for some $j$. If follows that we can express $f$  as $f_{0}(x)=\xi+\sum_{J}tr(x^{2^{j}+1})$ where $J$ is a subset of $\N$ with $J\not=\lbrace 0\rbrace$ and $\xi\in \F_{2}$.
Now, it is easy to chek that $x^{2^{j}+1}+(x+1)^{2^{j}+1}=x^{2^{j}}+x+1$ and thus $tr(x^{2^{j}+1})+tr((x+1)^{2^{j}+1})=1$.
Therefore $f_{0}(x)+f_{0}(x+1)=\sum_{J}1$ and $f_{0}(x)+f_{0}(x+1)=0$ if  $|J|$ is even and $f_{0}(x)+f_{0}(x+1)=1$ if $|J|$ is odd.
In both cases $D_{1}f_{0}$ is a constant function then  Part 1) is  a consequence of Theorem \ref{thm1}.

2) Since $f_{0}+ f_{1}=tr$ and $D_{1}f_{0}$ is a constant function then Theorem \ref{thm5} applies.

3) This is a direct consequence of Theorem \ref{thm7}.
\end{proof}

\begin{remark}
 Observe that for any $(2t)$-boolean function $G$ whose components are $g_{0}$ and  $g_{1}$ the TVR of $G$ $\phi_{G}(x,y)$ can be rewritten as
$$
\phi_{G}(x,y)=y(g_{0}(x)+g_{1}(x))+g_{0}(x).
$$
If $G=F$ then $\phi_{F}(x,y)=y\,tr(x)+f_{0}(x)$.
Since the degree of $f_{0}$ is 2 we deduce that the degree of $F$ also is 2.
Similarly for $\tilde{F}$: $\phi_{\tilde{F}}(x,y)=y(\tilde{f}_{0}(x)+\tilde{f}_{1}(x))+\tilde{f}_{0}(x)$.
We know from \cite[Lemma 2.5]{r6} that if the degree of $F$ is 2 then  $\tilde{F}$  also has degree 2.
From $\phi_{\tilde{F}}(x,y)=y(\tilde{f}_{0}(x)+\tilde{f}_{1}(x))+\tilde{f}_{0}(x)$ we deduce that degree of  $\tilde{f}_{0}$  is 2.
It follows that for every bent function among $F,\,\tilde{F},\,\bar{F}_{0},\,\bar{F}_{1},\,\tilde{\bar{F}}_{0},\,\tilde{\bar{F}}_{1}$, the degree is 2 and then the components have degree 2.
Obviously, $\bar{F}_{0}=\tilde{F}$, $\tilde{\bar{F}}_{0}=F$.
It is shown in \cite{r7} that if $\deg f_{0}=2$ there exists $e$ in $\F_{2^{2t-1}}$ such that $f_{0}(x)+ f_{1}(x)=tr(ex)$. Another way to prove 1) is  to show that $e=1$ when $\pi(x)$ is binary.
\end{remark}

\begin{ex}
\medskip\par\noindent
$t=4$.
\medskip\par\noindent
$F:\quad f_{0}(x)=tr(x^{3}+x^{9}),\,f_{1}(x)=f_{0}(x)+tr(x)$.
\medskip\par\noindent
$\tilde{F}:\quad \tilde{f}_{0}(x)=tr(x^{9}+x),\,\tilde{f}_{1}(x)=\tilde{f}_{0}(x)+tr(x)=tr(x^{9})$.
\medskip\par\noindent
$\bar{F}_{0}=\tilde{F}$, $\tilde{\bar{F}}_{0}=F$.
\medskip\par\noindent
$\bar{F}_{1}:\quad \bar{f}\,_{0}^{(1)}(x)=\tilde{f}_{1}(x),\,\bar{f}\,_{1}^{(1)}(x)=\tilde{f}_{1}(x)+tr(x)=\tilde{f}_{0}(x)$.
\medskip\par\noindent
$\tilde{\bar{F}}_{1}:\quad \tilde{\bar{f}}\,_{0}^{(1)}(x)=tr(x+x^{3}+x^{9}),\,\tilde{\bar{f}}\,_{1}^{(1)}(x)=\tilde{\bar{f}}\,_{0}^{(1)}(x)+tr(x+1)$.
\medskip\par\noindent
\end{ex}

\subsection{Comments}\label{sec2.4}\

$\bullet$ Note that starting from a near-bent function $f$ whose derivative $D_{1}f$ is a constant function and applying Theorem \ref{thm1} and Theorem \ref{thm7}, we are in position to construct six bent functions.

$\bullet$  It is easy to check that, for every boolean function $f$,
$$
D_{1}(f+1)=D_{1}(f),\qquad D_{1}(f+tr)=D_{1}(f)+1,
$$
if $f(0)=1$ then $(f+1)(0)=0$.

Now assume that $f$ is a near-bent function. As it will be shown by $(R_{2})$ in Section \ref{sec3.1.2}, $f+1,\,f+tr$ and $f+tr+1$  are also  near-bent functions.
From the above-mentioned results
$D_{1}(f+1)=0$ or $D_{1}(f+tr)=0$.

If $f(0)=0$ then $(f+1)(0)=1$ and $(f+1+tr)(0)=0$.
We deduce that among $f,\,f+1,\,f+tr,\,f+tr+1$ there always exists a near-bent function $h$ such that $D_{1}(h)=0$ and $h(0)=0$.
Therefore, in order to apply Theorem \ref{thm1} it is sufficient to find  a near-bent function $f$ such that $D_{1}(f)=0$ and $f(0)=0$.

In this case, consider the following polynomial $p(X)=\sum_{i=0}^{2t-2}f(\alpha^{i})X^{i}$ where $\alpha$ is a primitive root of $\F_{2^{2t-1}}$.
As pointed out in \cite{r10}, this is the representation of a word of a special cyclic code of length $2t-1$ over $\F_{2}$ which depends on the degree of $f$.

$\bullet$ The map sending a bent function to its pseudo-dual is not injective. For example  the bent function
defined by  $f_{0}(x)=tr(x^7+x^{13}+x^{19}+x^{21})$  and  $f_{1}(x)=f_{0}(x)+tr(x)$ and the bent function such that
$f_{0}(x)=tr(x^3+x^{11})$ and  $f_{1}(x)=f_{0}(x)+tr(x)$  have different duals, but the same pseudo-dual.

$\bullet$  Starting from a bent function $F$ which fulfill $(\mathcal{C})$ or $(\mathcal{T})$ it is possible to find other bent functions with the same properties.
The  bent functions $F$, the dual $\tilde{F}$, the pseudo-duals $\bar{F}_{0}$ and $\bar{F}_{1}$ and the duals  of these pseudo-duals
could be either distinct or not. The examples in  Section 4 show different  situations.

$\bullet$  In a very interesting paper \cite{r7} by Leander and McGuire the authors  consider the two-variable representations of  boolean functions
 with other notations. In particular, they introduce a characterization of near-bent functions $f$ such that $f$ and $f+tr$ are the components of
 a bent function (Theorem 3, $e=1$). This could be used to obtain an alternative proof of Theorem \ref{thm1}.

\section{Proofs}\label{sec3}
\subsection{Preliminaries}\label{sec3.1}
\subsubsection{Notation}\label{sec3.1.1}
Let $F$ be an $m$-boolean function.
\begin{itemize}
  \item The weight of  $F$ is defined by $w(F)=\sharp\lbrace v\in \F_{2}^{m}\mid F(v)=1\rbrace$.
  \item The TVR of $F$ is defined as in Section 1 by $$\phi_{F}(x,y)=(y+1)f_{0}(x)+yf_{1}(x).$$
  \item The Fourier transform of $F$ is defined as in Section 2.
  \item $T_{v}$ denotes the linear form of $\F_{2}^{m}$ defined by $T_{v}(X)=<v,X>$.
  \item If $m=2t-1$ then $tr$ denotes the trace function of $\F_{{2}^{2t-1}}$ and the map $x\rightarrow tr(ax)$ is denoted by $t_{a}$.
\end{itemize}

\subsubsection{Elementary and known results}\label{sec3.1.2}
We begin by summarizing some of  the elementary or classical results on bent and near-bent functions (see \cite{r1,r2,r9}).
\begin{enumerate}
  \item[$(R_{1})$] $\hat{F}(v)=2^{m}-2w(F+T_{v})$.
  \item[$(R_{2})$] Let $F$ be a $m$-boolean function and let $L$ be an affine linear form of $\F_{2}^{m}$.
                   $F$ is a bent function if and only if $F+L$ is a bent function.
                   $F$ is a near-bent function if and only if $F+L$ is a near-bent function.
  \item[$(R_{3})$] $w(F)=w(f_{0})+w(f_{1})$.
  \item[$(R_{4})$] $F$ is bent if and only if $\forall\, V\in \F_{2^{2t}}\, D_{V}F$ is balanced. That is
                   $$\sharp\lbrace U \in \F_{2^{2t}} \mid D_{V}F(U)=1\rbrace = \sharp\lbrace U \in \F_{2^{2t}} \mid D_{V}F(U)=0\rbrace.$$
  \item[$(R_{5})$] Let $\tilde{F}$ be the dual of a $(2t)$-bent function $F$. Then $\tilde{F}(v)=1$ if and only if $\hat{F}(v)=-2^{t}$.
  \item[$(R_{6})$] If $m=2t-1$ and $f$ is a near-bent function then the distribution of Fourier coefficients is as follows
\begin{align*}
&\hat{f}(v)=2^{t}\quad{\rm number\ of\ } v: 2^{2t-3}+(-1)^{f(0)}2^{t-2}, \\
&\hat{f}(v)=0\quad{\rm number\ of\ } v: 2^{2t-2}, \\
&\hat{f}(v)=-2^{t}\quad{\rm number\ of\ } v: 2^{2t-3}-(-1)^{f(0)}2^{t-2}.
\end{align*}
\end{enumerate}

\begin{comt}
$(R_{1})$, $(R_{2})$ and $(R_{5})$ follow immediatly from the definitions.
 $(R_{3})$  is obtained with straightforward calculations and $(R_{4})$ is classical.
The distribution given in $(R_{6})$ is a special cases of Proposition 4 in \cite{r2}.
\end{comt}

\subsubsection{Representation of $(2t)$-linear forms}\label{sec3.1.3}
The purpose of this part is to express linear forms and the inner product $< , >$ used in the calculation of the Fourier coefficients
 in such  a way which is consistent with the decomposition of $\F_{2^{2t}}$ as $\F_{2^{2t-1}}\times \F_{2}$.

\begin{definition}\label{def10}
For every $(a,\eta)$ in  $\F_{2^{2t}}$ the  map  $L_{(a,\eta)}$ from $\F_{2^{2t}}$ into $\F_{2}$ is defined by
$$L_{(a,\eta)}{(x,\nu)}=tr(ax)+\eta\nu.$$
\end{definition}

\begin{remark}
$L_{(a,\eta)}{(x,\nu)}$ is nothing but $tr(ax)+t_{r}^{(1)}(\eta\nu)$ where $t_{r}^{(1)}$ is the trace of $\F_{2}$ since
$t_{r}^{(1)}(\mu)=\mu$ for every $\mu$ in $\F_{2}$.

It can be easily checked that
\begin{description}
  \item[$(\star)$] $L_{(a,\eta)}$ is a linear form of $\F_{2^{2t}}$.
  \item[$(\star,\star)$] The map $(a,\eta)\longrightarrow L_{(a,\eta)}$ from $\F_{2^{2t}}$ to
                         $\lbrace  L_{(a,\eta)}\mid (a,\eta)\in \F_{2^{2t}}\rbrace$ is injective.
  \item[$(\star,\star,\star)$] The map $\big((a,\eta),(x,\nu)\big)\longrightarrow L_{(a,\eta)}{(x,\nu)}$
                            is a non-degenerate symetric bilinear form of $\F_{2^{2t}}$.
\end{description}
We immediately deduce from $(\star)$ and $(\star,\star)$ that
$\lbrace  L_{(a,\eta)}\mid (a,\eta)\in \F_{2^{2t}}\rbrace $ is  the set of linear forms of $\F_{2^{2t}}$.
On the other hand, $(\star,\star,\star)$  leads to a choice of the inner product  $T_{v}$ as defined in section \ref{sec2}.
\end{remark}

\begin{definition}\label{def11}
The inner product $< , >$ such that $T_{v}(X)=<v,X>$ is now defined by $T_{(a,\eta)}=L_{(a,\eta)}$. In other words,
$$<(a,\eta),(x,\nu)>=tr(ax)+\eta\nu.$$
\end{definition}

We immediatly deduce

\begin{proposition}\label{prop12}
 (*) $\quad \phi_{T_{(a,\eta)}}(x,y)=(y+1)tr(ax)+y(tr(ax)+\eta)$.
\medskip\par\noindent
 Let  $F$ be a $(2t)$-boolean function such that $\phi_{F}(x,y)=(y+1)f_{0}(x)+yf_{1}(x)$. Then
\medskip\par\noindent
(**)\ \ $\phi_{F+T_{(a,\eta)}}(x,y)=(y+1)(f_{0}(x)+tr(ax))+y(f_{1}(x)+tr(ax)+\eta)$.
\end{proposition}

\subsubsection{Representation of bent functions}\label{sec3.1.4}

As before, in this subsection we consider a $(2t)$-boolean $F$ function such that $\phi_{F}(x,y)=(y+1)f_{0}(x)+yf_{1}(x)$.

\begin{lemma}\label{lem13}\ \\
\hspace*{0.5cm} a) $\hat{F}(u,0)=\hat{f}_{0}(u)+\hat{f}_{1}(u)$.
\medskip\par\noindent
\hspace*{0.5cm} b) $\hat{F}(u,1)=\hat{f}_{0}(u)-\hat{f}_{1}(u)$.
\medskip\par\noindent
\hspace*{0.5cm} c) If $f_{0}+f_{1}=tr$ then $\hat{f}_{1}(u)=\hat{f}_{0}(u+1)$
\end{lemma}

\begin{proof}
From (**) and $(R_{3})$,
\begin{itemize}
  \item If $\eta=0$: $w(F+T_{(u,0)})=w(f_{0}+t_{u})+w(f_{1}+t_{u})$. According to $(R_{1})$ this means
        $$2^{2t-1}-\frac{1}{2}\hat{F}(u,0)=2^{2t-2}-\frac{1}{2}\hat{f}_{0}(u)+2^{2t-2}-\frac{1}{2}\hat{f}_{1}(u),$$
        and this leads to a).
  \item If $\eta=1$: first notice that $w(f_{1}+t_{u}+1)=2^{2t-1}-w(f_{1}+t_{u})$. Hence
        $$w(F+T_{(u,0)})=w(f_{0}+t_{u})+2^{2t-1}-w(f_{1}+t_{u}).$$
        By using $(R_{1})$ as above, we obtain the result of b).
  \item $\hat{f}_{0}(u)=2^{2t-1}-2w(f_{0}+t_{u})$ and $\hat{f}_{1}(u)=2^{2t-1}-2w(f_{0}+t_{u+1})$ whence $\hat{f}_{1}(u)=\hat{f}_{0}(u+1)$.\qedhere
\end{itemize}
\end{proof}

The next proposition is a version of a classical result  which can be found in several papers on bent functions
(\cite{r2,r9}). We give now a proof for sake of convenience.

\begin{proposition}\label{prop14}
$F$ is a bent function if and only if\\
\hspace*{0.5cm}  (a)  $f_{0}$ and $f_{1}$ are near-bent.\\
\hspace*{0.5cm}  (b)  $\forall a\in \F_{2^{2t-1}}\mid \hat{f_{0}}(a)\mid +\mid \hat{f_{1}}(a)\mid =2^t$.
\end{proposition}

\begin{remark}
(b) means that one of $\mid \hat{f_{0}}(a)\mid$ and $\mid \hat{f_{1}}(a)\mid$ is equal to $2^t$ and the other one is equal to 0.
\end{remark}

\begin{proof}
Let  $(a,\eta)$ be in $\F_{2^{2t}}$.

Assume $F$ is bent. From Lemma \ref{lem13}, $\hat{f_{0}}(a)=\frac{1}{2}\lbrack \hat{F}(a,0)+\hat{F}(a,1)\rbrack$ and
$\hat{f_{1}}(a)=\frac{1}{2}\lbrack \hat{F}(a,0)-\hat{F}(a,1)\rbrack$.
Since $F$ is bent, $\hat{f_{0}}(a)$ and $\hat{f_{0}}(a)$ are in $\lbrace -2^{t}, 2^{t}\rbrace$. By inspection of all possible case we see that
$\hat{f_{0}}(a)$ and $\hat{f_{1}}(a)$ are in $\lbrace -2^{t},0,-2^{t}\rbrace$ for every $a$, which means that $\hat{f_{0}}(a)$ and $\hat{f_{1}}(a)$ are near-bent. Furtheremore, we can check that in every case only one of $\hat{f_{0}}(a)$ and $\hat{f_{1}}(a)$ is $0$.

Conversely, now assume (a) and (b). By Lemma \ref{lem13}, this immediately  implies that for every $(a,\eta)$ in $\F_{2^{2t}}$ the weight of $\hat{F}(a,\eta)=\epsilon2^{t}$ with $\epsilon \in \lbrace -1,+1\rbrace$ and this proves that $F$ is bent.
\end{proof}

\subsection{A fundamental lemma}\label{sec3.2}
We need the following lemma which is  important  for the next  proofs and is a generalization of a classical result
on the Gold function (see \cite{r5}), since if $f$ is the Gold function then  $D_{1}f$ is a constant function.

\begin{lemma}\label{lem15}
Let $f$ be  a  $(2t-1)$-near-bent function.
\begin{itemize}
  \item If $D_{1}f=0$, then $\hat{f}(u)=0$ if and only if  $tr(u)=1$.
  \item If $D_{1}f=1$, then $\hat{f}(u)=0$ if and only if  $tr(u)=0$.
\end{itemize}
\end{lemma}

\begin{proof}
Assume that $D_{1}f=\omega$ with $\omega\in\F_{2}$ which means that $f(x+1)=f(x)+\omega$.
The transform $\tau:\,x\rightarrow x+1$ is a permutation of $\F_{{ }_{2^{2t-1}}}$ and then preserves the weight of every
$(2t-1)$-boolean function. Thus
$$
\sharp\lbrace x\mid f(x)+tr(ux)=1\rbrace=\sharp\lbrace x\mid f(x+1)+tr(u(x+1))=1\rbrace,
$$
$$(E)\quad \sharp\lbrace x\mid f(x)+tr(ux)=1\rbrace=\sharp\lbrace x\mid f(x)+\omega+tr(ux)+tr(u)=1\rbrace.
$$
If $tr(u)+\omega=1$  the right hand member of $(E)$ is
$$
\sharp\lbrace x\mid f(x)+tr(ux)=0\rbrace=2^{2t-1}-\sharp\lbrace x\mid f(x)+tr(ux)=1\rbrace.
$$
Hence $(E)$ becomes
$$
\sharp\lbrace x\mid f(x)+tr(ux)=1\rbrace=2^{2t-1}-\sharp\lbrace x\mid f(x)+tr(ux)=1\rbrace.
$$
In other words $w(f+t_{u})=2^{2t-1}-w(f+t_{u})$ and thus
\par\noindent
\hspace*{1cm}  If $tr(u)+\omega=1$: $w(f+t_{u})=2^{2t-2}$ which is equivalent to $\hat{f}(u)=0$.
\par\noindent
For $\omega=0$ or $\omega=1$ the number of $u$ such that $tr(u)+\omega=1$ is $2^{2t-2}$ and  $(R_{6})$ claims that this  is also the number of
$u$ such that $\hat{f}(u)=0$. Then, immediately $\hat{f}(u)=0$ if and only if  $tr(u)+\omega=1$.
Finally if $\omega=0$ then  $\hat{f}(u)=0$ if and only if  $tr(u)=1$ and if $\omega=1$ then  $\hat{f}(u)=0$ if and only if  $tr(u)=0$.
\end{proof}

\subsection{Proof of Theorem \ref{thm1}}\label{sec3.3}
Let $f_{0}$ be a $(2t-1)$-near bent function. Let $F$ be the $(2t)$-boolean function whose components are $f_{0}$ and $f_{1}$ such that $f_{1}=f_{0}+tr$. Our task is to prove that if $D_{1}f_{0}$ is a constant function, then $F$ is a bent function.

First notice that $f_{1}$ also is a near-bent function. This means that $\hat{f}_{0}$ and $\hat{f}_{1}$ take their values in $\lbrace -2^{t},0, 2^{t}\rbrace$.
Since $f_{1}=f_{0}+tr$ and  $a\in \F_{2^{2t-1}}$, according to Lemma \ref{lem13}, c): $\hat{f}_{1}(a)=\hat{f}_{0}(a+1)$.

Because $2t-1$ is odd, observe that $tr(1)=1$. Therefore, one element of $\lbrace tr(a),tr(a+1)\rbrace$ is $0$ and the other one is $1$.

Lemma \ref{lem15} shows that if $D_{1}f_{0}$ is a constant function then $\hat{f}_{0}(a)$ and $\hat{f}_{0}(a+1)$ are not $0$ in the same time.
Hence, one element of $\lbrace \hat{f}_{0}(a),\hat{f}_{1}(a)\rbrace$ is $0$ and the other one is $2^{t}$ or $-2^{t}$.

Finally, according to Proposition \ref{prop14}, this is the proof that $F$ is bent.

\subsection{Proof of Theorem \ref{thm2}}\label{sec3.4}
First, we need the following proposition

\begin{proposition}\label{prop16}
Let $f_{0}$ and $f_{1}$ be the components of  a bent function $F$. Let $\omega\in\F_{2}$. Then we have\\
\centerline{$D_{1}f_{0}=\omega$ if and only if  $D_{1}f_{1}=\omega +1$.}
\end{proposition}

\begin{proof}
$D_{(0,1)}F(X)=F(X+(0,1))+F(X)$. The TVR of $D_{(0,1)}F$ is
$$
(y+1)(f_{0}(x+1)+f_{0}(x))+y(f_{1}(x+1)+f_{1}(x))=(y+1)D_{1}f_{0}(x)+yD_{1}f_{1}(x).
$$
From $(R_{3}):\, w(D_{(0,1)}F)=w(D_{1}f_{0})+w(D_{1}f_{1})$.
Furthermore, $(R_{4})$ shows that $w(D_{(0,1)}F)=2^{2t-1}$. Thus
$$
(\div)\quad 2^{2t-1}=w(D_{1}f_{0})+w(D_{1}f_{1}).
$$
On the other hand, if $f$ is any $(2t-1)$-boolean function, then $w(D_{1}f)=2^{2t-1}$ is equivalent to $D_{1}f=1$,
while $w(D_{1}f)=0$ is equivalent to $D_{1}f=0$. Therefore, $(\div)$ proves that if one of the two derivatives $D_{1}f_{0}$ and $D_{1}f_{1}$ is $0$, then the other one
 is $1$.
\end{proof}

Now we go back to the proof of Theorem \ref{thm2}.
$D_{1}\tilde{f}_{0}=0$ means  $\tilde{f}_{0}(u)=\tilde{f}_{0}(u+1)$ for all $u$ in $\F_{2^{2t-1}}$.
Since  $\tilde{f}_{0}$ is the restriction of $\tilde{F}$ to $\F_{2t}^{(0)}=\lbrace (u,0)\mid u \in \F_{{ }_{2^{2t-1}}}\rbrace$,
 then in order to prove that $D_{1}\tilde{f}_{0}=0$ it suffices to show that
$$
\forall \, u\in \F_{2^{2t-1}}:\tilde{F}(u,0)=\tilde{F}(u+1,0).
$$
Using Lemma \ref{lem13}, since $f_{0}+f_{1}=tr$, we have successively
\medskip\par\noindent
 \hspace*{0.5cm} $\hat{f}_{1}(u)=\hat{f}_{0}(u+1)$.
\medskip\par\noindent
 \hspace*{0.5cm} $\hat{F}(u,0)=\hat{f}_{0}(u)+\hat{f}_{1}(u)=\hat{f}_{0}(u)+\hat{f}_{0}(u+1)$.
\medskip\par\noindent
\hspace*{0.5cm} $\hat{F}(u+1,0)=\hat{f}_{0}(u+1)+\hat{f}_{1}(u+1)= \hat{f}_{0}(u+1)+\hat{f}_{0}(u+1)$.
\medskip\par\noindent
\hspace*{0.6cm}$\hat{F}(u,0)=\hat{F}(u+1,0)$
\medskip\par\noindent
Following $(R_{5})$ we deduce that $\tilde{F}(u,0)=1$ if and only if $\tilde{F}(u+1,0)=1$ for all $u\in \F_{2^{2t-1}}$ and then
$\tilde{F}(u,0)=\tilde{F}(u+1,0)$. $D_{1}\tilde{f}_{1}=1$ is a direct consequence of Proposition \ref{prop16}.

\subsection{Proof of Theorem \ref{thm4}}\label{sec3.5}
According to Lemma \ref{lem13}, for every $a$ in $\F_{2^{2t-1}}$,
$\hat{f}_{1}(a)=\hat{f}_{0}(a+1)$ whence $\hat{F}(a,0)=\hat{f}_{0}(a)+\hat{f}_{0}(a+1)$.
On the other hand, a remark after Proposition \ref{prop14} says that
one of $|\hat{f_{0}}(a)|$ and $|\hat{f_{1}}(a)|$ is equal to $2^t$ and the other one is equal to 0.
It follows that every $a$ in $\F_{2^{2t-1}}$ belongs to one of the following sets
\begin{align*}
&\mathcal{A}_{1}=\lbrace a\in \F_{2^{2t-1}}\mid\hat{f}_{0}(a)=-2^t\hbox{ and }\hat{f}_{0}(a+1)=0\rbrace,\\
&\mathcal{A}_{2}=\lbrace a\in \F_{2^{2t-1}}\mid\hat{f}_{0}(a)=0\hbox{ and }\hat{f}_{0}(a+1)=-2^t\rbrace,\\
&\mathcal{A}_{3}=\lbrace a\in \F_{2^{2t-1}}\mid\hat{f}_{0}(a)=2^t\hbox{ and }\hat{f}_{0}(a+1)=0\rbrace,\\
&\mathcal{A}_{4}=\lbrace a\in \F_{2^{2t-1}}\mid\hat{f}_{0}(a)=0\hbox{ and }\hat{f}_{0}(a+1)=2^t\rbrace.
\end{align*}
The definition of the dual of $F$ induces that $(a,\eta)$ is in the support of  $\tilde{F}$ if and only if $\hat{F}(a,\eta)=-2^t$.
If we notice
\begin{itemize}
  \item $\hat{F}(a,0)=-2^t$ if  $a\in \mathcal{A}_{1}$ or $a\in \mathcal{A}_{2}$;
  \item $\hat{F}(a,0)=2^t$ if  $a\in \mathcal{A}_{3}$ or $a\in \mathcal{A}_{4}$,
\end{itemize}
we deduce that, $(a,0)$ is in the support of  $\tilde{F}$ if and only if $a \in \mathcal{A}_{1}\cup \mathcal{A}_{2}$.
In other words  the support of $\tilde{f}_{0}$ is $\mathcal{S}_{0}=\mathcal{A}_{1}\cup \mathcal{A}_{2}$.

From the descriptions of $\mathcal{A}_{1}$ and $\mathcal{A}_{2}$, if $a\in \mathcal{A}_{1}$ then $a+1\in \mathcal{A}_{2}$ and if $b\in \mathcal{A}_{2}$ then $b=a+1$ with $a=b+1\in \mathcal{A}_{1}$.
Hence $\mathcal{A}_{2}=\lbrace u+1\mid u\in \mathcal{A}_{1}\rbrace$.
Finally, by inspection we see that  $\mathcal{A}_{1}$ is nothing but the set
$\mathcal{S}=\lbrace v\in \F_{2^{2t-1}}\mid \hat{f}_{0}(v)=-2^t \rbrace$
and this leads to  result  1).

Since $f_{0}+f_{1}=tr$ then by Lemma \ref{lem13}, $\hat{f}_{1}(u)=\hat{f}_{0}(u+1)$.
Hence, $(a,1)$ is in the support of  $\tilde{F}$ if and only if  $a$ is in $\mathcal{A}_{3}$ or in  $\mathcal{A}_{3}$.
Consequently, the support of $\tilde{f}_{1}$ is $\mathcal{T}_{1}=\mathcal{A}_{3}\cup \mathcal{A}_{4}$.
It can be easily seen that $\mathcal{T}_{1}$ is the symmetric difference of the support of $\tilde{f}_{0}$ and of
  $\mathcal{G}=\lbrace v\in \F_{2^{2t-1}}\mid \hat{f}_{0}(v)=0\rbrace$.
This immediatly gives result 2).

\subsection{Proof of Theorem \ref{thm5}}\label{sec3.6}
Since $h_{0}+h_{1}=tr$ then Theorem \ref{thm4} claims that  $\tilde{h}_{0}+\tilde{h}_{1}$ is the characteristic function of
 $\mathcal{H}=\lbrace v\in \F_{2^{2t-1}}\mid \hat{h}_{0}(v)=0\rbrace$. We know from Lemma \ref{lem15} that if  $D_{1}h_{0}=0$ then  $\hat{h}_{0}(u)=0$ is equivalent to
 $tr(u)=1$ and  if  $D_{1}h_{0}=1$ then  $\hat{h}_{0}(u)=0$  is equivalent to $tr(u)=0$. This the same as saying that the characteristic function of  $\mathcal{H}$ is
 $tr$ if $D_{1}h_{0}=0$ and is $tr+1$ if $D_{1}h_{0}=1$ and this is the expected result.

\subsection{Proof of Theorem \ref{thm7}}\label{sec3.7}
The components   of the considered boolean functions are
\begin{itemize}
  \item $f_{0}$,$f_{1}$ for $F$ and $\tilde{f}_{0}$,$\tilde{f}_{1}$ for the dual $\tilde{F}$ of  $F$.
  \item $\bar{f}\,_{0}^{(0)}$, $\bar{f}\,_{1}^{(0)}$ for the pseudo-dual $\bar{F}_{0}$ and $\bar{f}\,_{0}^{(1)}$, $\bar{f}\,_{1}^{(1)}$
        for the pseudo-dual $\bar{F}_{1}$.
  \item $\tilde{\bar{f}}\,_{0}^{(0)}$, $\tilde{\bar{f}}\,_{1}^{(0)}$ for the dual $\tilde{\bar{F}}_{0}$ of $\bar{F}_{0}$
        and $\tilde{\bar{f}}\,_{0}^{(1)}$, $\tilde{\bar{f}}\,_{1}^{(1)}$ for the dual $\tilde{\bar{F}}_{1}$ of $\bar{F}_{1}$.
\end{itemize}

Proof of A):
 Since $f_{0}+f_{1}=tr$ then $D_{1}\tilde{f}_{0}=0$ and $D_{1}\tilde{f}_{1}=1$  (Theorem \ref{thm2}).
We deduce from Theorem \ref{thm1} that $\bar{F}_{0}$ and  $\bar{F}_{1}$ are bent functions.

Proof of B) and C):
From the definitions of the duals,
$$
\bar{f}\,_{0}^{(0)}=\tilde{f}_{0},\,\bar{f}\,_{0}^{(1)}=\tilde{f}_{1}, \,\bar{f}\,_{0}^{(0)}+\bar{f}\,_{1}^{(0)}=tr,\,\bar{f}\,_{0}^{(1)}+\bar{f}\,_{1}^{(1)}=tr.
$$
Hence, according to Theorem \ref{thm2},
$D_{1}\tilde{\bar{f}}\,_{0}^{(0)}=0$ and $D_{1}\tilde{\bar{f}}\,_{0}^{(1)}=0$ and thus  $\tilde{\bar{F}}_{0}$ and $\tilde{\bar{F}}_{1}$ meet $(\mathcal{C})$.

Now applying Theorem \ref{thm5} to  $\bar{F}_{0}$ with $h_{0}=\bar{f}\,_{0}^{(0)}$ and $h_{1}=\bar{f}\,_{1}^{(0)}$, we get $h_{0}+h_{1}=tr$.
Since  $D_{1}\bar{f}\,_{0}^{(0)}=0$ then
$\tilde{h}_{0}+\tilde{h}_{1}=\tilde{\bar{f}}\,_{0}^{(0)}+\tilde{\bar{f}}\,_{1}^{(0)}=tr$. That is $\tilde{\bar{F}}_{0}$ meets $(\mathcal{T})$ with  $\xi=0$.

Similarly, for $\bar{F}_{1}$ with $h_{0}=\bar{f}\,_{0}^{(1)}$ and $h_{1}=\bar{f}\,_{1}^{(1)}$ we have $h_{0}+h_{1}=tr$.
Remark $D_{1}\bar{f}\,_{0}^{(1)}=D_{1}\tilde{f}_{1}$ whence, again from Theorem \ref{thm2}, $D_{1}\bar{f}\,_{0}^{(1)}=1$ and Theorem \ref{thm5}
gives $\tilde{\bar{f}}\,_{0}^{(0)}+\tilde{\bar{f}}\,_{1}^{(0)}=tr+1$ and thus $\tilde{\bar{F}}_{1}$ meets $(\mathcal{T})$ with  $\xi=1$.

\section{Examples}\label{sec4}
If $(\mathcal{C})$ or $(\mathcal{T})$ hold for a bent function $F$,
the  bent functions $F$, the dual $\tilde{F}$, the pseudo-duals  $\bar{F}_{0}$,  $\bar{F}_{1}$ and the duals  of these pseudo-duals
can be distinct or not. The examples below show different  situations.

For every following example, condition $(\mathcal{T})$ is satisfied for the initial bent function $F$.
$t=4$ and $tr$ is the trace function of $\F_{2^{7}}$.

\begin{exam}\label{ex1}
\medskip\par\noindent
$F:\ f_{0}(x)=tr(x^7+x^{13})$ not $(\mathcal{C})$, $(\mathcal{T})$\\
\hspace*{1cm} $f_{1}(x)=f_{0}(x)+tr(x)$
\medskip\par\noindent
$\tilde{F}:\quad \tilde{f}_{0}(x)=tr(x^5+x^7+x^9+x^{13}+x^{19}+x^{21})\quad $ \\
\hspace*{1cm} $\tilde{f}_{1}(x)=\tilde{f}_{0}(x)+tr(x+x^5+x^9)$
\medskip\par\noindent
$\bar{F}_{0}:\quad \bar{f}\,_{0}^{(0)}(x)=\tilde{f}_{0}(x),\,\bar{f}\,_{1}^{(0)}(x)=\tilde{f}_{0}(x)+tr(x)$
\medskip\par\noindent
$\bar{F}_{1}:\quad \bar{f}\,_{0}^{(1)}(x)=\tilde{f}_{1}(x),\,\bar{f}\,_{1}^{(1)}(x)=\tilde{f}_{1}(x)+tr(x)$
\medskip\par\noindent
$\tilde{\bar{F}}_{0}:\quad \tilde{\bar{f}}\,_{0}^{(0)}(x)=tr(x+x^7+x^9+x^{13}+x^{19}+x^{21}),\quad $
\\
\hspace*{1cm} $\tilde{\bar{f}}\,_{1}^{(0)}(x)=\tilde{\bar{f}}\,_{0}^{(0)}(x)+tr(x)$
\medskip\par\noindent
$\tilde{\bar{F}}_{1}:\quad \tilde{\bar{f}}\,_{0}^{(1)}(x)=tr(x+x^3+x^7+x^{13}+x^{19}+x^{21})$
\\
\hspace*{1cm} $\tilde{\bar{f}}\,_{1}^{(1)}(x)=\tilde{\bar{f}}\,_{0}^{(1)}(x)+tr(x+1)$
\end{exam}

\begin{exam}\label{ex2}
\medskip\par\noindent
$F:\ f_{0}(x)=tr(x^{15}+x^{27}+x^{29}+x^{43})$, not $(\mathcal{C})$, $(\mathcal{T})$\\
\hspace*{1cm} $f_{1}(x)=f_{0}(x)+tr(x)$
\medskip\par\noindent
$\tilde{F}:\quad \tilde{f}_{0}(x)=tr(x+x^3+x^5+x^9)\quad $ \\
\hspace*{1cm} $\tilde{f}_{1}(x)=\tilde{f}_{0}(x)+tr(x^5+x^7+x^{11}+x^{19}+x^{21})$
\medskip\par\noindent
$\bar{F}_{0}:\quad \bar{f}\,_{0}^{(0)}(x)=\tilde{f}_{0}(x),\,\bar{f}\,_{1}^{(0)}(x)=\tilde{f}_{0}(x)+tr(x)$
\medskip\par\noindent
$\bar{F}_{1}:\quad \bar{f}\,_{0}^{(1)}(x)=\tilde{f}_{1}(x),\,\bar{f}\,_{1}^{(1)}(x)=\tilde{f}_{1}(x)+tr(x)$
\medskip\par\noindent
$\tilde{\bar{F}}_{0}=\bar{F}_{0}$
\medskip\par\noindent
$\tilde{\bar{F}}_{1}:\quad \tilde{\bar{f}}\,_{0}^{(1)}(x)=tr(x+x^3+x^5+x^7+x^9+x^{11}+x^{19}+x^{21})$
\\
\hspace*{1cm} $\tilde{\bar{f}}\,_{1}^{(1)}(x)=\tilde{\bar{f}}_{0,1}(x)+tr(x+1)$
\end{exam}

\begin{exam}\label{ex3}
$F:\ f_{0}(x)=tr(x+x^3+x^7+x^{11}+x^{19}+x^{21})$, $(\mathcal{C})$, $(\mathcal{T})$\\
\hspace*{1cm} $f_{1}(x)=f_{0}(x)+tr(x)$
\medskip\par\noindent
$\tilde{F}:\quad \tilde{f}_{0}(x)=tr(x^7+x^{11}+x^{19}+x^{21})\quad$\\
\hspace*{1cm} $\tilde{f}_{1}(x)=\tilde{f}_{0}(x)+tr(x)$
\medskip\par\noindent
$\bar{F}_{0}=\tilde{F}$ and $\tilde{\bar{F}}_{0}=F\quad \bar{F}_{1}=\tilde{F}$ and $\tilde{\bar{F}}_{1}=F$:
\end{exam}

\begin{exam}\label{ex4}
$F:\ f_{0}(x)=tr(x^3+x^5+x^7+x^{11}+x^{19}+x^{21})$, $(\mathcal{C})$, $(\mathcal{T})$\\
\hspace*{1cm} $f_{1}(x)=f_{0}(x)+tr(x)$
\medskip\par\noindent
$\tilde{F}=F=\bar{F}_{0}=\bar{F}_{1}=\tilde{\bar{F}}_{0}=\tilde{\bar{F}}_{1}$
\end{exam}

\subsection{A special example and an open question}
In \cite{r7} the authors recall the definition of non-weakly-normal bent function and they introduce (Fact 13) an example
of such a function $F$ in dimension 12 defined by $\phi_{F}(x,y)=(y+1)tr(x^{241}+x)+ytr(x^{241})$
 where $tr$ is the trace of $\F_{2^{11}}$ over $\F_{2}$.
We see that $f_{1}+f_{0}=tr(x)$ and thus Theorem \ref{thm7} holds.

An interesting open question is: Are $\tilde{F}$, $\bar{F}_{0},\,\bar{F}_{1},\,\tilde{\bar{F}}_{0},\,\tilde{\bar{F}}_{1}$
also non-weakly normal?

\section{Conclusion}\label{sec5}
We have introduced  a way to construct bent functions starting from a near-bent functions which fulfill the special condition of Theorem \ref{thm1}. Applying this Theorem and Theorem \ref{thm7} we obtain six bent functions. An open question now is to describe explicitely the near-bent functions $f$ such that $D_{1}f$ is a constant function, for example by means of the trace function.

Another question is to express the characteristic function of the set $\mathcal{S}$ which appears in Theorem \ref{thm4}, by using the trace function.

The results of this work could probably be generalized  by replacing
$1$ with $e$ such that $tr(e)=1$ and $tr(x)$ with $tr(ex)$ in $(\mathcal{C})$ and $(\mathcal{T})$.

\medskip

Received December 2011; revised November 2013.

\medskip

{\it E-mail address:} wolfmann@univ-tln.fr

\end{document}